\RequirePackage{pdf14}
\documentclass{statsoc}
\usepackage[a4paper,margin=0.8in,bottom=1in,top=0.9in]{geometry}

\usepackage{amsmath,amssymb,mathtools}
\usepackage{booktabs,xcolor}
\newtheorem{lemma}{Lemma}
\usepackage[ruled,norelsize,vlined]{algorithm2e}
\usepackage{graphicx}
\usepackage{subfigure}

\usepackage{lineno}

\usepackage{fancyhdr,amsmath,amssymb}
\RequirePackage[colorlinks,citecolor=blue,urlcolor=blue]{hyperref}
\usepackage{natbib}
\usepackage{etoolbox}
\makeatletter
\patchcmd{\@makecaption}
 {\parbox}
  {\advance\@tempdima-\fontdimen2} 
  {}{}
\makeatother  

\usepackage{mhequ}
\newcommand{\be}{\begin{equs}}
\newcommand{\ee}{\end{equs}}

\usepackage[normalem]{ulem}

\newcommand{\X}{X}
\newcommand{\Sb}{S}
\newcommand{\U}{U}

\DeclareMathOperator*{\argmin}{arg\,min} 
\DeclareMathOperator*{\argmax}{arg\,max} 
\DeclareMathOperator{\cov}{cov}
\newcommand{\z}{Z}

\title[Removing the influence of a groups variable]{Removing the influence of a group variable in high-dimensional \\ predictive modelling}

\author{Emanuele Aliverti}
\address{Department of Statistical Sciences, University of Padova}
\email{aliverti@stat.unipd.it}

\author{Kristian Lum}
\address{Human Rights Data Analysis Group, San Francisco}
\author{James E. Johndrow}
\address{Department of Statistics, Stanford University}
\author[Aliverti, Lum, Johndrow and Dunson]{David B. Dunson}
\address{Department of Statistical Science, Duke University}
\begin{document}

\date{}
\begin{abstract}
In many application areas, predictive models are used to support or make important decisions. There is increasing awareness that these models may contain spurious or otherwise undesirable correlations. Such correlations may arise from a variety of sources, including batch effects,  systematic measurement errors, or sampling bias.
Without explicit adjustment, machine learning algorithms trained using these data can produce poor out-of-sample predictions which propagate these undesirable correlations.
We propose a method to pre-process the training data, producing an adjusted dataset that is statistically independent of the nuisance variables with minimum information loss.
We develop a conceptually simple approach for creating an adjusted dataset in high-dimensional settings based on a constrained form of matrix decomposition. 
The resulting dataset can then be used in any predictive algorithm with the guarantee that predictions will be statistically independent of the group variable.
We develop a scalable algorithm for implementing the method, along with theory support in the form of independence guarantees and optimality.
The method is illustrated on some simulation examples and applied to two case studies: removing machine-specific correlations from brain scan data, and removing race and ethnicity information from a dataset used to predict recidivism. That the motivation for removing undesirable correlations is quite different in the two applications illustrates the broad applicability of our approach.

\end{abstract}


\section{Introduction}
\label{sec:1}

As machine learning models are deployed in high-stakes application areas, there has been growing recognition that machine learning algorithms can reproduce unwanted associations encoded in the data upon which they were trained  \citep[e.g.][]{dunson2018statistics, zech2018confounding}. 
Such potentially problematic associations are typically between the outcome of interest $Y$, which is specific to the scientific application under investigation, and one or more group variables $Z$.  A model that retains correlations between $Y$ and $Z$ may be undesirable for different reasons depending on the application area. 
This article focuses on a method for adjusting training data to prevent models fit to that data from learning associations between undesirable features and the outcome of interest. Without proper adjustment, the association between the outcome of interest and unwanted group information can obscure scientifically relevant signals leading to biased results and misleading conclusions, or perpetuate socially undesirable stereotypes or practices.

In this article we focus on two challenging research application areas where this problem is particularly salient. The first is the field of ``omic'' data,  where there is strong interest in using algorithms to relate images, scans, and medical records to outcomes of interest, such as disease and subject-specific phenotypes. A recent example highlights how unwanted associations can manifest in these settings. \citet{zech2018confounding} showed that radiographic image data encoded information on the specific hospital system from which the data were collected, since different systems tended to use different imaging equipment. The systematic differences in the image data produced by the different equipment were not readily apparent to the human eye. 

However,  when the data were used to train a machine learning model for pneumonia screening,  the model learned to associate these equipment-specific characteristics with the outcome of interest. That is, the model leveraged the fact that some hospital systems  saw higher rates of disease and also used equipment that left a specific signature in the imaging data. The model's predictions were then partially based on correlations between the signature of the machines used in the higher-prevalence hospital systems and the disease. The model's reliance on such associations jeopardizes the generalizability of the results, as there is no reason to believe that other higher prevalence hospitals will use similar equipment in the future. This issue is also detrimental for in-sample evaluation and inferential purposes, since regarding such associations as risk factors for disease is clinically misleading.

In the field of medical imaging and biostatistics, removal of spurious effects in high dimensional data has been an important problem for at least the last decade.
Various other tools for dealing with this issue have been developed, with methods based on matrix factorisation being very popular in genomics \cite[e.g.][]{wall2003singular}.
For example, \citet{alter2000singular} use a singular value decomposition (\textsc{svd}) of the original expression matrix to filter out the singular vector associated with systematic and non-experimental variations, and reconstruct an adjusted expression matrix with the remaining singular vectors. A similar strategy was pursued in \citet{bylesjo2007orthogonal} and \citet{leek2007}, where matrix decompositions are used to separate the latent structures of the expression matrix that are not affected by batch assignment from the nuisance batch variable. Model-based methods have also been successful. For example, distance discrimination based on support vector machines \citep{benito2004adjustment} and multilevel models estimated via empirical Bayes \citep{johnson2007} have been developed to deal with groups with few observations. For a more comprehensive review of approaches for batch correction and array normalization see \citet{luo2010comparison}, \citet{lazar2012}, and references therein.

The focus of these methods has mostly been on pre-processing high-dimensional data in order to remove the batch effects and then conducting different analysis on the adjusted data; for example, visualising the data to locate clusters of similarly expressed genes \citep[e.g.][]{lazar2012}. Data visualisation is often used as the main evaluation tool to determine the success in removing batch-effects, while assessment and evaluation for out-of-sample data has not received much interest \citep{luo2010comparison}. We take an approach that is similar to the methods based on matrix decomposition, with the focus being on providing a low-dimensional approximation of the data matrix which guarantees that batch effects are removed with minimum information loss. Unlike the approaches popular in the biostatistics literature, our approach explicitly focuses on quantitatively addressing the problem of removing unwanted associations by imposing a form of statistical independence among the adjusted data and the group membership variable and explicitly enforcing such constraints in the matrix decomposition. In addition, our approach focuses on the generalization problem, guaranteeing that predictions for \emph{future} data will be statistically independent from group variables.

Our second motivating example comes from algorithmically automated or assisted decision-making, where algorithms are deployed to aid or replace human decision-making. Examples include algorithms for selecting short-listed candidates for job positions or college admission. 
If the measurement error varies with group membership status, associations between the group variable and the outcome of interest may obscure correlations with scientific factors of genuine interest.  Moreover, when groups are defined by socially sensitive attributes -- such as race or gender -- a model that learns these  associations may lead to unfair or discriminatory predictions. 
For example, there has been much recent attention in the ``algorithmic fairness'' literature on the use of criminal risk assessment models, many of which use demographic, criminal history, and other information to predict whether someone who has been arrested will be re-arrested in the future. These predictions then inform decisions on pre-trial detention, sentencing, and parole. In many cases the model seeks to predict re-arrest, which may be influenced by racially biased policing \citep{bridges1988law, rudovsky2001law, simoiu2017problem}. When risk assessment models are trained using these data, the end result is that individuals in racial minority groups tend to be systematically assigned to higher risk categories on average \citep[][]{johndrow2019algorithm,angwin2016machine}. 
Even when information on race is not explicitly included as a covariate, this phenomenon persists, since other variables that are included in the model are strongly associated with race. 

Models for which one group is more likely to be flagged as ``high risk" are in conflict with one notion of algorithmic fairness-- statistical parity. However, many definitions are available in the literature. See \citet{berk2018, mitchell2018prediction, corbett2018measure} 
for recent overviews of mathematical notions of fairness. 
Arguably, each definition proposed in the literature can be regarded as a valid concept of algorithmic fairness, though which notion is most sensible depends heavily on the data and application. Some methods for achieving notions of fairness have focused on altering the objective function used in estimation to, for example, balance the false positive or negative rate across groups \citep[e.g.][]{hardt2016,zafar2017}. Others have focused on defining metrics that more fairly characterize relevant differences between individuals to enforce notions of individual fairness. 
The recent work of \cite{zhang2018mitigating} achieves a variety of notions of fairness through adversarial learning--- to achieve demographic parity, the adversary is tasked with inferring group membership status from model's predictions. This method is applied to recidivism prediction in \cite{wadsworth2018achieving}. 

Other approaches to creating ``fair'' models modify the training data rather than the objective function, with some focusing on modifications to $Y$ \citep{kamiran2009classifying}. The approach we take most closely follows \citet{johndrow2019algorithm}, \citet{adler2018auditing} and \citet{feldman2015certifying} in pre-processing the covariates $X$, guaranteeing that any model estimated on the ``adjusted'' data will produce out-of-sample predictions with formal guarantees of statistical independence. Our algorithm has significant advantages over existing approaches in scalability and ability to handle high-dimensional data efficiently. 

Motivated by datasets from these two distinct application areas, in this article we propose a simple method to adjust high-dimensional datasets in order to remove associations between covariates and a nuisance or otherwise undesirable variable. 
Our procedure creates an adjusted dataset with guarantees that algorithms estimated using the adjusted data produce predictions which will be statistically independent from the nuisance variable. The main advantage of our contribution is its simplicity and scalability to a large number of covariates ($p$), including when $p$ is greater than the number of observations $n$. In this sense, it is particularly well-suited for applications like brain imaging and ``omic'' data, in which the observed covariates are high-dimensional. It also has significant advantages in the case of highly collinear predictors, which is very common in applications. Moreover, the solution we propose has theoretical guarantees, both in terms of optimal dimension reduction and the ability to achieve independence from the undesirable variables under a linearity condition.
We also provide guarantees of minimal information loss during the pre-processing.

It is worth mentioning that the goal of achieving statistical independence is not without controversy.
In medical imaging predictions, one could imagine a scenario where hospitals use a similar equipment type and also treat similar patient populations.
Scrubbing the data of dependence on imaging equipment may also remove some information that is useful for prediction and is not captured by other covariates in the dataset.
A generous read of a common critique of creating race-independent data in the recidivism prediction setting is that race might be an important risk factor, since the correlation between race and the likelihood of re-arrest reflects the reality that people of color are more likely to experience conditions leading to criminality, such as poverty.
Creating race-independent predictions may {\it under}-estimates the likelihood of re-arrest, even after accounting for systematic bias in the process by which individuals are arrested.
In reality, it is difficult to know how much of the difference in re-arrest rate can be attributed to differential patterns of enforcement relative to differential participation in criminal activity.
In this setting, it may be ethically appealing to choose equal group-wise treatment, which is achieved by enforcing independence between predictions and race. Furthermore, even if a reasonable ground truth were available, one might still want to create race-independent predictions to avoid further exacerbating the existing racial disparities in incarceration rates. 

\section{Data description}

The first case study discussed in this article comes from medical imaging data, and is drawn from a study in neuroscience conducted by the Human Connectome Project (\textsc{hcp}) on $n=1056$ adult subjects with no history of neurological disorder \citep{glasser2016human,glasser2013minimal}. The study provides, for each individual, information on the structural interconnections among the $68$ brain regions characterizing the Desikan atlas \citep{desikan:2006}, measured through a combination of diffusion and structural magnetic resonance imaging; see \citet{zhang2018relationships} for additional details.
Many different features are also available, covering a wide range of biographical, physiological and behavioural information at the individual level and technical information related to the specific session in which brain scans were collected. 
For an extended description of the Humane Connectome Project, the tools involved in the collection process and the aims of the study see \citet{zhang2018relationships} and \citet{glasser2016human,glasser2013minimal}.
For our purposes, it is enough to characterize the outcomes of interest as physiological and behavioural traits and the covariates as data on the presence and strength of connections between the $68$ brain regions.

Recent developments in neuroscience have stimulated considerable interest in analysing the relationship among brain structure or activity and subject-specific traits, with the main focus being on detecting if variations in the brain structure are associated with variation in phenotypes \citep[e.g.][]{genovese2002thresholding,zhang2018relationships,Durante:2018}.
There is evidence for the existence of brain differences across subjects with severe drug addictions, both in terms of functional connectivity \citep{wilcox2011enhanced, kelly2011reduced} and volumes of brain regions \citep{beck2012effect,goldstein2009neurocircuitry}. As discussed in the Introduction, a fundamental problem with observational medical data is the presence of spurious associations such as batch effects. In neuroscience, it is  well known that  subject motion, eye movements, different protocols and hardware-specific features, among many others can complicate data analysis \citep[e.g.][]{basser2002diffusion,sandrini2011use}. Here, we use the \textsc{hcp} data to investigate the relationship between connectivity and drug use while removing batch effects from the identity of the imaging equipment used for each subject. 

The second case study discussed in this article comes from algorithmic criminal risk assessment. We will focus here on the \textsc{compas} dataset, a notable example in the fairness literature which includes detailed information on criminal history for more than $6000$ defendants in Broward County, Florida. 
For each individual, several features on criminal history are available, such as the number of past felonies, misdemeanors, and juvenile offenses; additional demographic information include the sex, age and race of each defendant.
Defendants are followed over two years after their release, and the data contain an indicator for whether each defendant is re-arrested within this time range.
See \citet{larson2016we} for details on the data collection process and additional information on the dataset.
The focus of this example is on predicting two-year recidivism as a function of defendant's demographic information and criminal history. 
As discussed in the Introduction, algorithms for making such predictions are routinely used in courtrooms to advise judges, and concerns about the fairness of such tools with respect to race of the defendants were recently raised, mainly using this specific dataset as an example \citep{angwin2016machine}.
Therefore, it is of interest to develop novel methods to produce predictions while avoiding disparate treatment on the basis of race.

\section{Generating data orthogonal to groups}
\subsection{Notation and setup}
Let $X$ denote an $n \times p$ data matrix of $p$ features measured over $n$ subjects, and let $Z$ denote an additional nuisance variable; for example, brain scanner in the neuroscience case study or race in the recidivism example. 
We focus for simplicity on a scalar $Z$, but the methods directly generalize to multivariate nuisance variables.
We seek to estimate $\widetilde \X$, an $n \times p$ reconstructed version of the data matrix that is orthogonal to $Z$ with minimal information loss. In our setting, the reconstructed version is used to produce a prediction rule $\hat Y(\widetilde X)$ that returns a prediction $\hat Y$ of $Y$ for any input $\tilde X$. 
Our aim is to guarantee that $\hat Y(\widetilde X)$ is statistically independent of $Z$.

We will focus on statistical models linear in the covariates, such as generalised linear models or support vector machines with linear kernels. It is easy to check that when $\hat Y(\widetilde X)$ is a linear function of $\widetilde X$, $\cov( \widetilde X, Z ) = 0$ implies $\cov(\hat Y, Z ) = 0$. Our procedure transforms $X$ by imposing orthogonality between $\widetilde \X$ and $Z$, while attempting to preserve as much of the information in $\X$ as possible. The former requirement is geometrically analogous to requiring that $Z$ is in the null space of $\widetilde \X$, 
so that it is not possible to predict $Z$ using the transformed variables in a statistical model which is linear in the covariates. Non-linear dependencies could still be present in the transformed matrix $\widetilde X$, so that predictions of models which are not linear in the covariates would still depend on $Z$. One potential solution to this problem is to include interactions in the matrix $X$ before transformation; we use this approach in the recidivism case study in Section \ref{sec:COMPAS}. Another possible solution is to attempt to remove nonlinear dependence as well, as in \cite{johndrow2019algorithm}. We do not pursue the latter here, favoring the simplicity and computational scalability of imposing a zero covariance constraint. Indeed, one of our two motivating applications is a large $p$, small $n$ problem, for which fast computation is critical.

Another desirable property of our procedure is dimensionality reduction. In high-dimensional settings, it is often assumed that covariate matrices have approximately a low-rank representation. We express a reduced rank approximation of $X$ as $\widetilde\X =SU^T$, where $\U$ is a $p \times k$ matrix of $k$ linear orthonormal basis vectors and $\Sb$ is the $n \times k$ matrix of associated scores. 
The problem of preprocessing the data to ensure Orthogonality to Groups (henceforth \textsc{og}) can be expressed as a minimization of the Frobenius distance between the original data and the approximated version, $\|\X - \widetilde \X\|^2_F$, under the constraint $\langle\widetilde\X,Z\rangle=0$. 
With $\widetilde \X = S U^T$, this leads to the following optimization problem.
\begin{equation}
\label{eq:min}
\argmin_{\Sb,\U}  \;\|\X - \Sb\U^T\|^2_F, \quad
\text{subject to}\; \langle \Sb \U^T,Z  \rangle = 0, \quad  \U \in \mathcal{G}_{p,k}
\end{equation}
where $\mathcal{G}_{p,k}$ is the Grassman manifold consisting of orthonormal matrices \citep{james1954normal}. While we focus on a univariate categorical $Z$ in this article, the procedure can be used as-is for continuous $Z$, and extension to multivariate $Z$ is straightforward. 

Since the constraints are separable, it is possible to reformulate Equation~\ref{eq:min} as $p$ distinct constraints, one over each column of $\widetilde \X$.
Moreover, since any column of $\widetilde \X$ is a linear combination of the $k$ columns of $\Sb$, and $\U$ is orthonormal, the $p$ constraints over $\widetilde \X$ can be equivalently expressed as $k$ constraints over the columns of the score matrix $\Sb$.
The matrix $\U$ is forced to lie on the Grassman manifold to prevent degeneracy, such as basis vectors being identically zero or solutions with double multiplicity.
The optimization problem admits an equivalent formulation in terms of Lagrange multipliers,
\begin{equation}
\label{eq:prob}
\argmin_{\Sb,\U}\bigg\{ \frac {1}{n}\sum_{i=1}^n \|x_i - \sum_{j=1}^k s_{ij} u_j^T\|^2 + \frac{2}{n}\sum_{j=1}^k\lambda_j \sum_{i=1}^n s_{ij}z_i\bigg\},
\end{equation}
with the introduction of the factor $2/n$ for ease of computation.

\subsection{Theoretical support}
\label{sec:fpl}
The following Lemma characterizes the solution of the \textsc{og} optimization problem, which can be interpreted as the residual of a multivariate regression among left singular values and a group variable. Let $V_k\Sigma_k\U_k^T$ denote the rank-$k$ Singular Values Decomposition (\textsc{svd}) of $X$.
\begin{lemma}
\label{lemma1}
The problem stated in Equation~\ref{eq:min} can be solved exactly, and admits an explicit solution in terms of singular values. The solution is equal to $\widetilde X = (I_n - P_z)V_k\Sigma_k\U_k^T$, with $P_Z = Z(Z^TZ)^{-1}Z^T$.
\end{lemma}

Proofs are given in Appendix~\ref{app:proof}. The computational cost of the overall procedure is dominated by the cost of the truncated \textsc{svd}, which can be computed with modern methods in $O(nk^2)$ \citep{golub2012matrix}. The procedure outlined in Lemma~\ref{lemma1} is simple and only involves matrix decomposition and least squares theory; hence we can fully characterise the solution and its properties. In the univariate case, the expression in Lemma \ref{lemma1} is easily interpretable since $P_z$ and $(I_n - P_z)$ are $n\times n$ matrices with elements $$[P_z]_{ij}=\frac{1}{n} +\frac{z_iz_j}{\sum_{i=1}^n z_i^2}, \qquad [(I_n - P_z)]_{ij}= \mathbf{I}[i=j] - \frac{1}{n} - \frac{z_iz_j}{\sum_{i=1}^n z_i^2}, \quad i,j=1,\dots,n. $$
	
 The following Lemma guarantees that the solution $\widetilde X$ of the \textsc{og} algorithm optimally preserves information in $X$.
\begin{lemma}
\label{lemma2}
The solution $\widetilde X$ of the \textsc{og} algorithm is the best rank-$k$ approximation, in Frobenius norm, of the data matrix $\X$ under the \textsc{og} constraint.
\end{lemma}

The \textsc{svd} achieves the minimum error in Frobenius distance among all matrices of rank-$k$ \citep[e.g.,][]{golub2012matrix}. Naturally, the introduction of additional constraints reduces the accuracy of the approximation relative to the \textsc{svd}. The following result bounds the additional error analytically.
\begin{lemma}
\label{lemma3}
Let $\widetilde X_k = V_k D_kU_k^T$ denote the best rank-$k$ approximation of the matrix $X$ obtained from the truncated \textsc{svd} of rank $k$. The reconstruction error of the \textsc{og} algorithm is lower bounded by the optimal error rate of $\widetilde X_k$, and the amount of additional error is equal to $\|P_zV_kD_k\|_F^2$.
\end{lemma}
The additional reconstruction error can be interpreted as a measure of the collinearity between the subspace spanned by $Z$ and the left singular vectors of the data $X$. The more correlated the singular vectors are with the group variable, the greater the additional error relative to the solution without the \textsc{og} constraint.
When $Z$ is in the null space of $X$, the solution is identical to the truncated singular value decomposition and the reconstruction achieves the minimum error.
Therefore, if the correlation between the group variable and the observed data is negligible, the procedure achieves a loss of information which is close to the \textsc{svd}. 
\subsection{Sparse \textsc{og} procedure}
\label{sec:sog}
Estimation of low-dimensional structure from high-dimensional data can be challenging when the number of features $p$ is larger than the number of observations $n$.
In this setting, common methods in multivariate analysis impose constraints on the elements of a matrix decomposition, usually through an $\ell_1$-norm penalty to favour sparsity and improve numerical estimation \citep[e.g.][]{zou2006sparse,jolliffe2003modified, witten2009}. 

To make the \textsc{og} problem tractable and stable when the number of features is very large -- potentially larger than the number of observations -- we introduce additional constraints in the algorithm. We will build our method on a standard procedure to perform sparse matrix decomposition \citep[e.g.][Chapter 8]{hastie2015statistical}, and adapt the computations to introduce the orthogonality constraint. We define the Sparse Orthogonal to Group (\textsc{sog}) optimization problem as follows.
\begin{equation}
\label{eq:sfpl}
\begin{split}
\argmin_{S,U} &\left\|X - SU^T \right\|_F^2 \\
\text{subject to}\quad \|u_j\|_2 \leq 1, \|u_j\|_1 \leq t,& \;\|s_j\|_2 \leq 1, s_j^Ts_l=0, \;s_j^T\z=0,
\end{split}
\end{equation}
for $j=1,\dots,k$ and $l\neq j$.
The problem in Equation~\ref{eq:sfpl} includes sparsity constraints over the vectors $u_j$ and imposes orthogonality constraints among the score vectors $s_j$ and the group variable, since the reconstruction $\widetilde \X = \Sb\U^T$ is a linear combination of the vectors $s_j$.

To fix ideas, we focus initially on a rank-$1$ approximation. Adapting the results in \citet{witten2009}, it is possible to show that the solutions in $s$ and $u$ for Equation~\ref{eq:sfpl} when $k=1$ also solve
\begin{equation}
\label{eq:equivalent}
\argmax_{s,u}\; s^T X u\quad \text{subject to}\quad \|u\|_2 \leq 1, \|u\|_1 \leq t, \;\|s\|_2 \leq 1,  \;s^T\z=0.
\end{equation}
Although the minimisation in Equation~\ref{eq:equivalent} is not jointly convex in $s$ and $u$, it can be solved with an iterative algorithm. Since the additional orthogonality constraints do not involve the vector $u$, when $s$ is fixed the minimisation step is mathematically equivalent to a sparse  matrix decomposition with constraints on the right singular vectors. This takes the form
\begin{equation}
\label{eq:s}
\argmax_{u}\; b u\quad \text{subject to}\quad \|u\|_2 \leq 1, \|u\|_1 \leq t,
\end{equation}
with $b=s^TX$ and solution 
\be
u=g(b,\theta)=\frac{\mathcal{S}_\theta (b)}{\|\mathcal{S}_\theta(b)\|_2},  
\ee
where $\theta$ is a penalty term in the equivalent representation of the constrained problem in \eqref{eq:equivalent} in terms of Langrange multipliers, and $\mathcal{S}_\theta(x) := \text{sign}(x)(|x| -\theta)\mathbb{I}(|x|\geq \theta)$ is the soft threshold operator applied over every element separately. The value of $\theta$ is $0$ if $\|b\|_1\leq t$, and otherwise $\theta > 0$ is selected such that $\|g(b,\theta)\|_1= t$ \citep{hastie2015statistical,witten2009}.


When $u$ is fixed, the solution is similar to that described in Section~\ref{sec:fpl}, which can be seen by rearranging Equation~\ref{eq:equivalent} to obtain
\begin{equation}
\label{eq:u}
\argmax_{s}\; s^T a \quad \text{subject to}\quad\|s\|_2 \leq 1,  \;s^T\z=0,
\end{equation}
with $a=Xu$ \citep{witten2009}. The solution to Equation~\ref{eq:u} is directly related to the method outlined in Section~\ref{sec:fpl}, and is given by the following expression.
$$s=\frac{a-\beta Z}{\|a-\beta Z\|_2},$$
with $\beta = (Z^TZ)^{-1}Z^Ta$.

Solutions with rank greater than $1$ are obtained by consecutive univariate optimization.
For the $j$-th pair $(u_j,s_j)$, $j=2,\dots,k$, the vector $a$ in Equation~\ref{eq:u} is replaced with $P_{k-1}Xu_j^T$, where ${P_{k-1} = I_{n\times n} - \sum_{l=1}^{k-1} s_l s_l^T}$ projects $Xu_j^T$ onto the complement of the orthogonal subspace ${\text{span}( s_l,\dots, s_{k-1})}$, thus guaranteeing orthogonality among the vectors $s_j$, $j=1,\dots,k$. A more detailed description of the algorithm outlined above is given in Appendix~\ref{app:alg}.

\section{Simulation study}
We conduct a simulation study to evaluate the empirical performance of the proposed algorithms and compare them with different competitors.
The focus of the simulations is on assessing fidelity in reconstructing a high-dimensional data matrix, success in removing the influence of the group variable from predictions for future subjects, and evaluation of the goodness of fit of predictions.
We also compare our methods with two popular approaches developed in biostatistics; specifically, the \textsc{combat} method of \citet{johnson2007} and the \textsc{psva} approach of \citet{leek2007}.
These methods adjust for batch effects via Empirical Bayes and matrix decomposition, respectively.
The R code implementing the methods illustrated in this article is available at \url{github.com/emanuelealiverti/sog}, and it also includes a tutorial to reproduce the simulation studies. The approaches used as competitors are available through the R package \textsc{sva} \citep{sva2012}.

The first simulation scenario puts $n=1000$, $p=200$, and $X$ has rank $k=10$. The data matrix $X$ is constructed in two steps. We simulate a loading matrix $S$, with size $(n,k)$, and a score matrix $U$ with size $(k,p)$, with entries sampled from independent Gaussian distributions. A group variable $Z$ of length $n$ is sampled from independent Bernoulli distributions with probability equal to $0.5$.
Each $p$-dimensional row of the $(n\times p)$ data matrix $X$ is drawn from a $p$-variate standard Gaussian distribution with mean vector $\mu_i = (s_i -  \lambda z_i)U$, $i=1, \dots, n$ and $\lambda$ is sampled from a $k$-variate Gaussian distribution with identity covariance. Lastly, a continuous response $Y$ with elements $y_i$, $i=1,\dots,n$ is sampled by first generating the elements of $\beta$ independently from $\text{Uniform}(-5, 5)$, then sampling $y_i \sim N((s_i - \lambda z_i)\beta,1)$. We highlight that in this setting the data matrix $X$ has a low-rank structure, and the response variable $Y$ is a function both of the group variable $Z$ and the low-dimensional embedding of $X$. In the second simulation scenario, the construction of the data matrix $X$ is similar to the first setting, except that the response variable $Y$ does not depend on $Z$. This is achieved by sampling the elements $y_i$ of $Y$ from standard Gaussians with mean vector $\mu_i = s_i \beta$, $i=1, \dots, n$. Therefore, the response $Y$ depends on the data only though the low-dimensional embedding of $X$. The third scenario focuses on a ``large $p$ - small $n$'' setting, in which the dimension of the data matrix $X$ is $n=200,p=1000$ with $k=10$, and its construction follows the first scenario, with dimensions of the score and loading matrix modified accordingly. In the fourth and last scenario, the low-rank assumption is relaxed and we focus on a case with $k=p$, following the same generating processes described above.

In each scenario, data are divided into a training set and a test set, with size equal to $3/4$ and $1/4$ of the observations, respectively. Therefore, the number of observations in the training and test set is equal to $750$ and $250$ respectively in the first, second and fourth scenario, and is equal to $150$ and $50$ in the third.
In each scenario, the proposed procedures and the competitors are applied separately over the training set and the test set. 
Then, linear regressions are estimated on the adjusted training sets and used to provide predictions $\hat{Y}$ over the different test sets, where the performance is evaluated by comparing predictions with the truth.

\begin{table}
	\caption{Out-of sample predictions for $Y$. Lower values correspond to more accurate predictions.\label{tab:sim} Values are averaged over $50$ random splits, with standard deviations reported in brackets. Best performance is highlighted in boldface.}
	\fbox{%
\begin{tabular}{lrrrrc|c}
& & \textsc{combat} &\textsc{psva} & \textsc{og} & \textsc{sog} & \textsc{svd} \\
\midrule

\em Scenario 1 & \textsc{rmse} & 20.73 (3.78) & 17.49 (5.81) & 15.31 (3.30) & \textbf{13.19} (2.64) & 21.48 (4.44) \\
	       & \textsc{mae}  & 16.54 (3.00) & 13.63 (4.49) & 12.13 (2.54) & \textbf{10.45} (2.09) & 17.00 (3.40) \\
	       & \textsc{mdae} & 13.92 (2.67) & 11.39 (3.65) & 10.13 (2.14) & \textbf{8.69} (1.84)  & 14.08 (2.77) \\
	   \midrule            
	   \em Scenario 2 & \textsc {rmse}& 15.37 (2.49) & 13.99 (3.63) & 11.54 (1.85) & \textbf{10.36} (1.71) & 15.75 (2.34) \\
	       & \textsc {mae} & 12.28 (1.96) & 10.84 (2.95) & 9.25 (1.51)  & \textbf{8.28} (1.38)  & 12.60 (1.89) \\
	       & \textsc {mdae}& 10.47 (1.74) & 9.16 (2.75)  & 7.85 (1.41)  & \textbf{6.97} (1.19)  & 10.71 (1.72) \\
	   \midrule            
	   \em Scenario 3 & \textsc {rmse}& 25.78 (4.54) & 22.93 (8.65) & 18.87 (3.96) & \textbf{16.15} (2.91) & 26.38 (4.74) \\
	       & \textsc {mae} & 20.87 (3.78) & 18.34 (7.08) & 15.14 (3.31) & \textbf{12.99} (2.45) & 21.15 (4.01) \\
	       & \textsc {mdae}& 18.14 (3.88) & 15.54 (6.58) & 12.64 (3.34) & \textbf{11.01} (2.47) & 17.86 (4.10) \\
	   \midrule            
	   \em Scenario 4 & \textsc {rmse}& 44.74 (2.55) & 45.70 (2.29) & 42.20 (1.98) & \textbf{41.93} (2.03) & 45.12 (2.27) \\
	       & \textsc {mae} & 35.96 (2.18) & 36.61 (2.01) & 33.92 (1.65) & \textbf{33.76} (1.72) & 36.28 (1.91) \\
	       & \textsc {mdae}& 30.75 (2.45) & 31.20 (2.37) & 29.25 (2.30) & \textbf{29.06} (2.03) & 31.00 (2.29) \\
\end{tabular}
   }
\end{table}

Table~\ref{tab:sim} reports the root mean square error (\textsc{rmse}), mean absolute error (\textsc{mae}) and median absolute error (\textsc{mdae}) for the adjusted predictions from the competitors (\textsc{combat}, \textsc{psva}), the proposed methods (\textsc{og}, \textsc{sog}) and the unadjusted case (\textsc{svd}), in which linear regressions are estimated on the left singular vectors of $X$. In order to reduce simulation error in the results from splitting the data into training and test sets, results are averaged across $50$ different splits, with standard deviations across splits reported in brackets.
Empirical results suggest that models estimated on datasets adjusted with the competitor methods have comparable performances, with \textsc{psva} providing more accurate predictions than \textsc{combat} in most settings. 
Both \textsc{og} and \textsc{sog} outperform the competitors in all the settings considered.
For example, in the third scenario involving a $p \gg n$ setting, the \textsc{rmse} of the \textsc{psva} is $22.93$ compared to $16.15$ for \textsc{sog}, which also provides better results for the other metrics considered.
It is worth highlighting that even when $X$ is not low rank, as in the fourth scenario, our methods provide more accurate out-of-sample predictions than the competitors.

We also compare the accuracy of each method in recovering the data matrix $X$ in simulation scenario 1 with increasing values of the rank $k$ of the approximation $\widetilde{X}$. 
The rank of the reconstruction is directly specified for \textsc{og}, \textsc{sog} and \textsc{psva} during estimation.
Batch-adjustment via \textsc{combat}, instead, does not utilize a low-rank representation.
To create a fair basis for comparison, we use a rank-$k$ \textsc{svd} decomposition of the batch-corrected matrix obtained from \textsc{combat} in the comparison of reconstruction accuracy.
Table~\ref{tab:simX} illustrates the Frobenius norms between the training set of the first scenario and its approximations for increasing values of the rank of each approximation. 
Results are averaged over $50$ random splits, and standard deviations across splits are reported in brackets. 

For each value of the rank $k$, the best reconstruction is, as expected, obtained by the \textsc{svd} decomposition, which returns a perfect reconstruction for ranks greater or equal than the true value $k=10$, but of course \textsc{svd} does not impose orthogonality to $Z$, and thus is reported only to provide a basis for comparison of the reconstruction error of the other methods.
Among the procedures that do seek to remove the influence of $Z$, the best reconstruction (minimum error) is obtained by \textsc{og} for every value of $k$. This empirical result is consistent with Lemma~\ref{lemma2}, which illustrates the optimality of the method in terms of reconstruction error.  For values of $k$ greater than or equal to $10$, the error is stable and corresponds to the quantity derived in Lemma~\ref{lemma3}. 
The error achieved with the \textsc{sog} algorithm is considerably higher than that of the \textsc{og} algorithm. This is expected since the true singular vectors are not sparse in our simulation scenarios. Despite this, \textsc{sog} performs better than \textsc{og} in out of sample prediction, suggesting that the additional regularization is useful in prediction. 
Among the competitors, \textsc{combat} is more accurate in reconstructing $X$ than \textsc{psva}, which encounters estimation issues with values of the rank greater than $8$, as illustrated by the increasing standard deviation of the reconstruction error with large $k$. 



\begin{table}
	\caption{Reconstruction error of the data matrix in the first scenario, measured in Frobenius norm. Results are averaged across $50$ random splits, with standard deviations reported in brackets. Lower values represent more accurate reconstructions. \label{tab:simX}}
	\fbox{%
\begin{tabular}{ccccc|c}
Rank $k$ & \textsc{combat} &\textsc{psva} & \textsc{og} & \textsc{sog} & \textsc{svd} \\
  \midrule
2  & 310.48 (1.45) & 347.78 (2.97)   & 309.03 (1.48) & 338.01 (1.56) & 305.65 (1.31) \\
3  & 282.80 (1.36) & 326.58 (3.25)   & 281.07 (1.38) & 318.90 (1.49) & 277.32 (1.23) \\
4  & 256.61 (1.24) & 304.34 (4.70)   & 254.53 (1.27) & 300.93 (1.38) & 250.28 (1.07) \\
5  & 228.90 (1.35) & 286.01 (15.83)  & 226.37 (1.37) & 282.60 (1.71) & 221.56 (1.06) \\
6  & 202.38 (1.33) & 264.54 (18.54)  & 199.27 (1.33) & 266.01 (1.50) & 193.74 (1.00) \\
7  & 174.52 (1.31) & 247.20 (36.40)  & 170.67 (1.33) & 249.55 (1.42) & 164.13 (1.04) \\
8  & 143.73 (1.41) & 231.63 (58.65)  & 138.83 (1.41) & 233.48 (1.56) & 130.62 (1.13) \\
9  & 107.80 (1.62) & 291.11 (106.68) & 100.86 (1.63) & 217.29 (1.69) & 89.19 (1.11) \\
10 & 60.97 (2.85)  & 139.36 (187.70) & 47.09 (3.16)  & 201.88 (1.92) & 0.00 (0.00) \\
11 & 62.49 (2.82)  & 162.44 (192.83) & 47.09 (3.16)  & 201.64 (3.21) & 0.00 (0.00) \\
\end{tabular}
   }
\end{table}


Finally, we analyse success in removing information about $Z$ from predictions based on adjusted data. Figure~\ref{fig:sim} compares the empirical correlation between the out-of-sample predictions for $\hat Y$ and the group variable $Z$, over $50$ random splits, for the first simulation scenario.
Results from Figure~\ref{fig:sim} suggest a general tendency of all the adjustment methods to reduce the correlation between the out-of-sample predictions and group variable, confirming the ability of the proposed algorithms and the competitors to remove the linear effect of group information from predictions.
Results from the two competitor approaches suggest that predictions from \textsc{combat} have smaller correlation than \textsc{psva} on average, and results are also less variable across different splits.
The high variability observed with \textsc{psva}, and its very modest effect at reducing correlation with $Z$, might be due to the convergence issues described previously.
As expected, predictions from \textsc{og} and \textsc{sog} always have zero correlation across all splits, since this constraint is explicitly imposed in the algorithm during optimization.

\begin{figure}[tb]
\includegraphics[width = \textwidth]{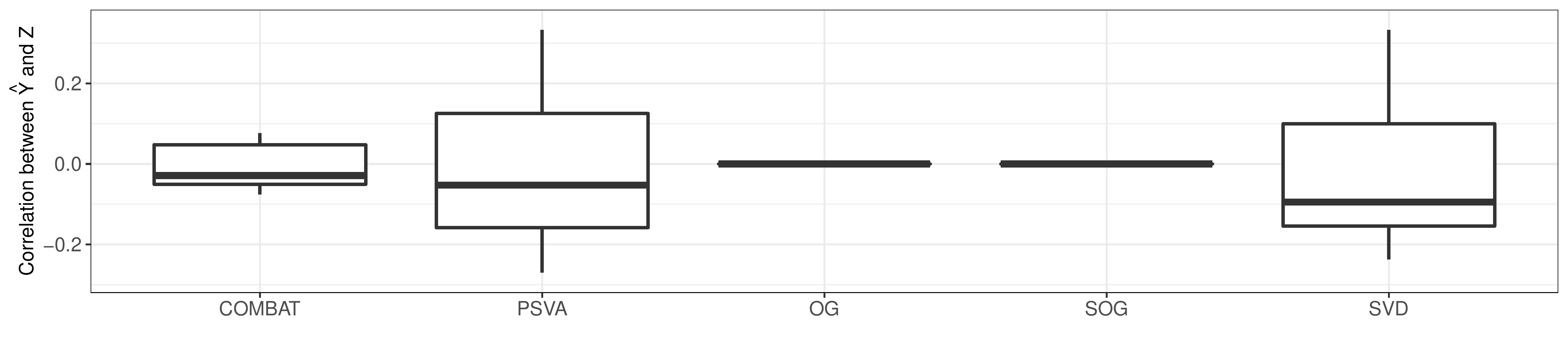}
\caption{Correlation between $\hat{Y}$ and $Z$ across $50$ random splits into training and test set in the first scenario.}
\label{fig:sim}
\end{figure}

\section{Application}
\label{sec:secAppl}
\subsection{Human connectome project} \label{sec:brains}
To apply \textsc{og} and \textsc{sog} to the imaging case study, the brains scans were first vectorised into a $n \times p$ matrix $\X$, with $n=1065$ subjects and $p=2278$ features corresponding to the strength of connection among all pairs of brain regions. The outcome of interest is coded as a binary variable $Y$ indicating a positive result to a drug test for at least one among Cocaine, Opiates, Amphetamines, MethAmphetamine and Oxycontine. The nuisance variable $Z$ indicates the machine on which the data were gathered. Data are randomly divided into a training and a test set, with size equal to $798$ and $267$ observations, respectively. In order to reduce simulation error from data splitting, results are averaged over $50$ different splits into training and test. Table~\ref{table:brains} illustrates the averages and the standard deviations of Accuracy (\textsc{acc}), Area Under the Roc Curve (\textsc{auc}), True Positive Rates (\textsc{tpr}), and True Negative Rates (\textsc{tnr}) for the out-of-sample predictions across splits. The threshold for predicting observations as positive corresponds to the proportion of positive observations in the training data, which is very small (around $3\%$ for both scanners).

To provide a basis for comparison, analyses are conducted using the original covariates without any adjustment. The first columns of Table~\ref{table:brains} represent results for a logistic regression using sparse \textsc{pca} with $k=30$ components as covariates.
The second and third columns compare predictive performance for Lasso (\textsc{lasso}) and  Random Forest (\textsc{rf}), using all the unadjusted available covariates. 
The right half of Table~\ref{table:brains} shows results for the adjusted procedures, with the columns \textsc{og} and \textsc{sog} giving results from a logistic regression estimated on $\tilde X$ obtained from our two methods, while \textsc{lasso,og} and \textsc{rf,og} give results for predictions from Lasso and Random Forest estimated after \textsc{og} adjustment.
Empirical findings suggest that predictive performance is not noticeably worse for models estimated on reconstructed data $\tilde X$.
In some cases, performance improves, while in other cases it declines. A similar phenomenon is observed in the recidivism example, suggesting that the loss of predictive performance after performing our pre-processing procedure can be minimal in applications.
Figure~\ref{fig:brains} provides additional results illustrating the empirical distribution of $\hat Y$ under the four approaches. Results suggests that predictions produced by \textsc{logistic} and \textsc{lasso} are different across two scans, with predicted probabilities from the first Scanner (black histogram) being on average greater. 
After pre-processing using our procedures, the empirical distribution of predicted probabilities of drug consumption becomes more similar, with results from predictions formed on \textsc{sog}-preprocessed data being close to equal across scanners.  
We also conducted sensitivity analysis for different values of the approximation rank ranging in $\{10,50,100\}$. The results were consistent with the main empirical findings discussed above.

\begin{table}
	\caption{Predictive performance on the \textsc{hcp} dataset. Higher values correspond to more accurate predictions. Best performance is highlighted in boldface.
	\label{table:brains}}
\fbox{%
\begin{tabular}{llrc|ccccc}
	   &\textsc{logistic} &\textsc{ lasso} &\textsc{ rf} &\textsc{ og  } &\textsc{ sog } &\textsc{ lasso, og} &\textsc{ rf, og }  \\
	  \midrule
	  \textsc{acc} & 0.60 (0.06)          & \textbf{0.62} (0.06) & 0.62 (0.06) & 0.72 (0.32) & \textbf{0.79} (0.23) & 0.53 (0.06)          & 0.58 (0.07) \\
	  \textsc{auc} & \textbf{0.54} (0.08) & 0.52 (0.08)          & 0.51 (0.10) & 0.51 (0.06) & 0.53 (0.08)          & \textbf{0.57} (0.10) & 0.52 (0.10) \\
	  \textsc{tnr} & 0.61 (0.06)          & \textbf{0.63} (0.06) & 0.63 (0.07) & 0.27 (0.34) & 0.20 (0.25)          & \textbf{0.43} (0.07) & 0.37 (0.07) \\
	  \textsc{tpr} & 0.29 (0.18)          & \textbf{0.40} (0.16) & 0.33 (0.19) & 0.55 (0.33) & \textbf{0.57} (0.28) & 0.42 (0.18)          & 0.53 (0.18) \\
\end{tabular}
}
\end{table}

\begin{figure}[t]
\includegraphics[width = \textwidth]{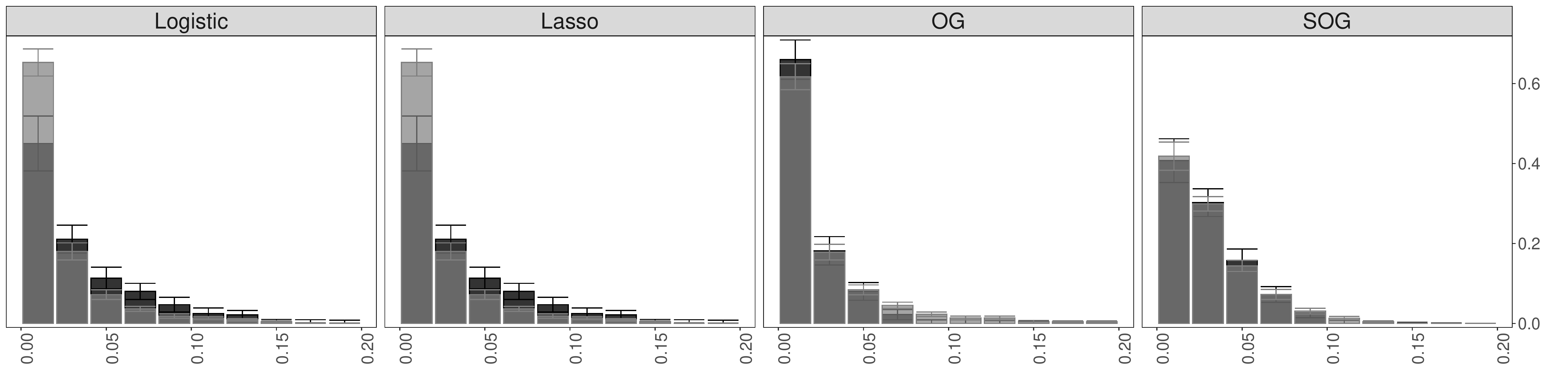}
\caption{Histograms for $\hat Y$ under the four approaches. Light gray corresponds to Scanner 1, black to Scanner 2. Regions where the two histograms overlap are shown as dark gray. The more extensive overlap between the histograms by scanner in the right two panels indicates success of the method at removing scanner information from the data. }
\label{fig:brains}
\end{figure}

\subsection{\textsc{compas} recidivism data}
\label{sec:COMPAS}
We apply our methods to the recidivism dataset by first constructing a design matrix which includes the available features and all the interaction terms, for a total of $p=64$ variables. 
The proportion of individuals with a positive outcome was roughly $40\%$ for Caucasian and $50\%$ for non-Caucasian.
Although the number of features is not particularly large in this case study, and considerably smaller than in the previous example, methods for pre-processing $X$ that require manual interaction with the statistician, such as \citet{johndrow2019algorithm}, are burdensome for even 64 covariates, since each covariate must be modeled separately conditional on others.
Moreover, the inclusion of every interaction term induces substantial collinearity in the design matrix, so that low-rank approximation of the data matrix is likely to be both accurate and improve computational and predictive performance for models estimated on the data.
Data are randomly divided into a training set and a test set, with size $5090$ and $1697$ observations, respectively. 
Table~\ref{table:COMPAS} reports Accuracy (\textsc{acc}), Area Under the Roc Curve (\textsc{auc}), True Positive Rates (\textsc{tpr}), True Negative Rates (\textsc{tnr}), Positive Predicted Values (\textsc{ppv}) and Negative Predicted Values (\textsc{npv}) of the out-of-sample predictions. Results are averaged over $50$ random splits of the data, with standard deviations reported in brackets.

The first and second columns of Table~\ref{table:COMPAS} represent, respectively, results for a logistic regression using all the available original covariates plus all interaction terms, and all the covariates and interaction terms plus race, while the third column corresponds to a Random Forest trained on the original data. 
The second part of \ref{table:COMPAS} illustrates results for the adjusted procedures, with \textsc{og} and \textsc{sog} reporting results from a logistic regression estimated on the adjusted data and \textsc{rf,og} from a Random Forest estimated on data preprocessed using \textsc{og}. As observed in the case of the brain scan data, predictive performance on the whole is quite similar when data are first pre-processed using our procedures, indicating that in some applications there is an almost ``free lunch;'' information on $Z$ can be removed with very little effect on the overall predictive performance on average.

Figure~\ref{fig:COMPAS} assesses success in removing dependence on race from $\hat Y$ by comparing the empirical cumulative distribution function (\textsc{cdf}) of the predicted probabilities with and without pre-processing of $X$. The first two panels of Figure~\ref{fig:COMPAS} represent the empirical \textsc{cdf} of $\hat{Y}$ for the models denoted as \textsc{logistic} and \textsc{logistic,z} in Table~\ref{table:COMPAS}. Results are averaged across splits, with bands corresponding to the $0.975$ and $0.025$ quantiles reported as dotted and dashed lines to illustrate variability of the curves across splits.
Without adjustment, white defendants (gray curves) are assigned lower probabilities of recidivism. This issue affects predictions both excluding the race variable (first panel) and including it is a predictor (second panel).
The third and fourth panels correspond, respectively, to predictions obtained from logistic regressions estimated on the data pre-processed with the \textsc{og} and \textsc{sog} procedures. 
The gap between the empirical \textsc{cdf}s is substantially reduced, both with the standard \textsc{og} and the sparse implementation, leading to distributions of the predictive probabilities which are more similar across different racial groups. This indicates that the procedure was largely successful at rendering predictions independent of race.

\begin{table}
	\caption{Predictive performance on the \textsc{compas} dataset. Higher values correspond to more accurate performances. Best performance is highlighted in boldface.
	\label{table:COMPAS}} 
\fbox{%
\begin{tabular}{llcc|cccccc}
    &\textsc{logistic} &\textsc{logistic, z} & \textsc{ rf} &\textsc{og} &\textsc{sog} &\textsc{rf, og}  \\
\toprule      
\textsc{acc} & 0.669 (0.02) & \textbf{0.671} (0.01) & 0.671 (0.01)          & \textbf{0.654} (0.01) & 0.654 (0.01) & 0.606 (0.01) \\
\textsc{auc} & 0.712 (0.02) & \textbf{0.714} (0.02) & 0.712 (0.01)          & \textbf{0.708} (0.01) & 0.708 (0.01) & 0.642 (0.01) \\
\textsc{tnr} & 0.719 (0.02) & 0.716 (0.04)          & \textbf{0.766} (0.02) & 0.659 (0.02) & \textbf{0.661} (0.02) & 0.608 (0.01) \\
\textsc{tpr} & 0.609 (0.05) & \textbf{0.617} (0.03) & 0.558 (0.02)          & \textbf{0.649} (0.01) & 0.647 (0.01) & 0.602 (0.02) \\
\textsc{ppv} & 0.644 (0.02) & 0.645 (0.02)          & \textbf{0.665} (0.02) & 0.613 (0.01) & \textbf{0.614} (0.01) & 0.562 (0.02) \\
\textsc{npv} & 0.689 (0.02) & \textbf{0.692} (0.01) & 0.675 (0.01)          & \textbf{0.692} (0.01) & 0.691 (0.01) & 0.647 (0.01) \\
\end{tabular}}
\end{table}

\begin{figure}[h]
\includegraphics[width = \textwidth]{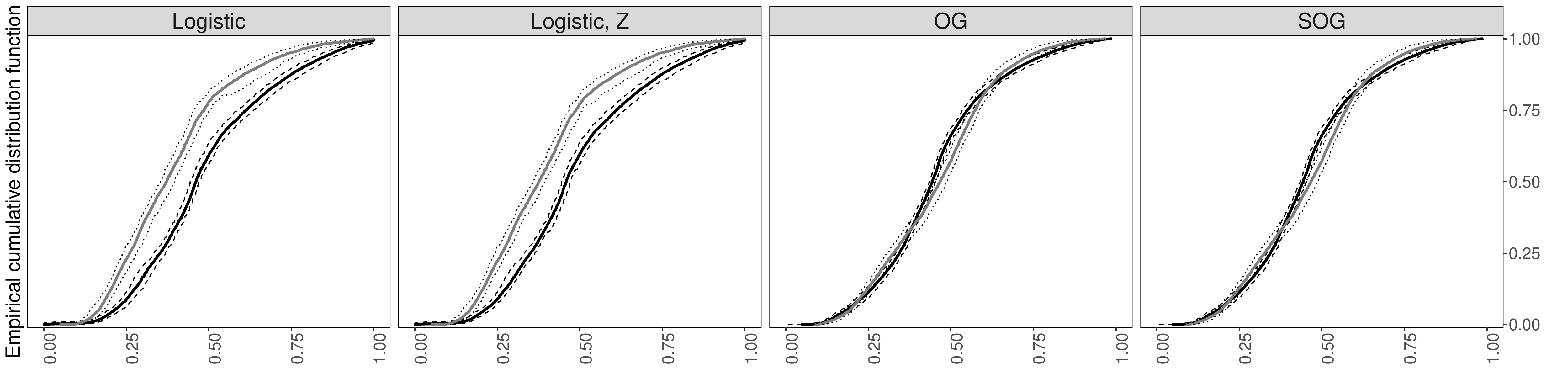}
\caption{Empirical cumulative distribution functions for $\hat Y$ under the four approaches. Gray refers to white ethnicity, black to non-white. Dashed and dotted lines corresponds to confidendence bands.}
\label{fig:COMPAS}
\end{figure}

\section{Discussion}
We have proposed a simple approach for adjusting a dataset so that predictive algorithms applied to the adjusted data produce predictions that are independent of a group variable.  This is motivated by two very different applications, involving adjusting for batch effects in neuroimaging and removing the influence of race/ethnicity in risk assessments in criminal justice.  Although these applications areas are fundamentally different, our adjustment procedure can be used in both cases.  An appealing aspect is that the procedure is agnostic to the type of predictive algorithm used, so that a single adjusted dataset can be released without prior knowledge of the types of predictive algorithms that will be applied.

There are several interesting areas for future research building upon the proposed framework.  One interesting direction is to consider more complex types of covariates.  For example, in neuroimaging one may obtain a tensor- or function-valued predictor.  For functional predictors, it is potentially straightforward to modify the proposed framework by leveraging developments in functional principal components analysis (\textsc{fpca}), combining \textsc{fpca} factorizations with our proposed algorithm.  For tensor predictors, one can similarly rely on tensor \textsc{pca}-type factorizations. 
However, in such settings, it is important to think carefully about the ramifications of orthogonality assumptions, as unlike for matrices most tensors are not orthogonally decomposable. Another important direction is to consider more complex data structures; for example, containing longitudinal observations, a network-type dependence structure and potentially censoring.  In the risk assessment setting, data are increasingly collected over time and there is dependence across different individuals which should be accommodated.

\section*{Acknowledgments}
This work was partially supported by grant ``Fair Predictive Algorithms'' from the Laura \& John Arnold Foundation.

\appendix
\section{Appendix}
\subsection{Results on \textsc{og} procedure}
\label{app:proof}
\begin{proof}[Proof of Lemma~\ref{lemma1}]
Focusing on the case $k=1$, the approximation of the original set of data consists of finding the closest rank-$1$ matrix (vector).
Equation~\ref{eq:prob} can be reformulated as

\begin{align}
\label{eq:1d}
\argmin_{s_1,u_1}&\bigg\{ \frac {1}{n}\sum_{i=1}^n \|x_i - s_{i1} u_1^T\|^2 + \frac{2}{n}\lambda_1 \sum_{i=1}^n s_{i1}z_i\bigg\},
\end{align}

and some algebra and the orthonormal condition on $u_1$ allows us to express the loss function to be minimized as
\begin{align*}
L(s_1,u_1) &=\frac {1}{n}\sum_{i=1}^n (x_i - s_{i1} u_1^T)^T (x_i - s_{i1} u_1^T) + \frac{2}{n}\lambda_1 \sum_{i=1}^n s_{i1}z_i \\
&=\frac {1}{n}\sum_{i=1}^n (x_i^Tx_i - 2s_{i1}x_iu_1^T+s_{i1}^2)+\frac{2}{n}\lambda_1 \sum_{i=1}^n s_{i1}z_i.
\end{align*}
\noindent
The function is quadratic, and the partial derivative in $s_{i1}$ leads to
\begin{align*}
\frac{\partial}{\partial s_{i1}}L(s_1,u_1) = \frac{1}{n}(- 2x_iu_1^T+2s_{i1}) +\frac{2}{n}\lambda_1 z_i,
\end{align*}
with stationary point given by $s_{i1} = x_iu_1^T-\lambda_1 z_i$.
The optimal score for the $i$-th subject is obtained by projecting the observed data onto the first basis, and then subtracting $\lambda_1$-times $z$.
The constraint does not involve the orthonormal basis $u_1$, hence the solution of Equation~\ref{eq:1d} for  $u_1$ is equivalent to the unconstrained scenario.
A standard result of linear algebra states that the optimal $u_1$ for Equation~\ref{eq:1d} without constraints equivalent to the first right singular vector of $\X$, or equivalently to the first eigenvector of the matrix $\X^T\X$ \citep[e.g.,][]{hastiefriedman}.
Plugging in the solution for $u_1$ and setting the derivative with respect to $\lambda_1$ equal to $0$ leads to

\begin{equation}
\sum_{i=1}^n (x_i u_1^T-\lambda_1 z_i) ^T z_i  = 0 \qquad \lambda_1 = \frac{\sum_{i=1}^n x_i u_1^T z_i}{\sum_{i=1}^n z_i^2} = \frac{\langle Xu_1^T,z \rangle}{\langle z,z\rangle},   
\end{equation}

a least squares estimate of $\X{u}_1^T$ over $z$.

\noindent
Consider now the more general problem formulated in Equation~\ref{eq:prob}. 
The derivatives with respect to the generic element $s_{ij}$ can be calculated easily due to the constraint on $\U$, which simplifies the computation.
Indeed, the optimal solution for the generic score $s_{ij}$ is given by 
\begin{equation}
\label{eq:derk}
s_{ij} = x_iu_j^T-\lambda_j z_i,
\end{equation} 
since $u_i^Tu_j=0$ for all 
$i \neq j$ and $u_j^Tu_j = 1$ for $j=1,\dots,k$.
The solution has an intuitive interpretation, since it implies that the optimal scores for the $j$-th dimension are obtained projecting the original data over the $j$-th basis, and then subtracting $\lambda_j$-times the observed value of $z$.
Moreover, since the \textsc{og} constraints do not involve any vector $u_j$, the optimization with respect to the basis can be derived from known results in linear algebra. The optimal value for the vector $u_j$, with $j=1,\dots,k$, is equal to the first $k$ right singular values of $\X$, sorted accordingly to the associated singular values \citep[e.g.,][]{bishop2006pattern,hastie2015statistical}.

The global solution for $\lambda = (\lambda_1, \dots, \lambda_k)$ can be derived from least squares theory, since we can interpret Equation~\ref{eq:derk} as a multivariate linear regression where the $k$ columns of the projected matrix $\X\U^T$ are response variables and $z$ a covariate.
The general optimal value for $\lambda_k$ is then equal to the multiple least squares solution $$\lambda_k = \frac{\langle\X u_k^T,z\rangle}{\langle z,z\rangle}.$$ 
\end{proof}
\begin{proof}[Proof of Lemma~\ref{lemma2}]
Since the optimization problem of Equation~\ref{eq:min} is quadratic with a linear constraint, any local minima is also a global minima. The solution performed via the singular value decomposition and the least squares constitute a stationary point, that is also global minimum.
\end{proof}

\begin{proof}[Proof of Lemma~\ref{lemma3}]
	Let $V_kD_kU_k^T$ define the rank-$k$ truncated \textsc{svd} decomposition of the matrix $\X$, using the first $k$ left and right singular vectors, and the first $k$ singular vales.
Let $\widetilde X_{OG}$ define the approximated reconstruction obtained by the \textsc{og} algorithm.
The reconstruction error between the original data matrix $X$ and its low-rank approximation $\widetilde X_{OG}$ can be decomposed as follow.
\begin{eqnarray*}
\|\X - \widetilde X_{OG}\|^2_F  =&\|\X - (V_k D_k-Z\lambda )U_k^T\|^2_F \\ 
=&\|\X - V_kD_kU_k^T+Z\lambda U_k^T\|^2_F\\
=&\|\X - V_kD_kU_k^T\|_F^2 + \|Z\lambda U_k^T\|^2_F
 +&2\langle\X - V_kD_kU_k^T,Z\lambda U_k^T \rangle_F.
\end{eqnarray*}
The Frobenius-inner product term is equal to $0$ due to the orthogonality of the singular vectors, and rearranging terms the following expression is obtained.
\begin{eqnarray*}
	\|\X - \widetilde\X_{OR}\|^2_F - \|\X - V_kD_kU_k^T\|_F^2 = &\|Z\lambda U^T\|^2_F=\|Z\lambda \|_F^2,
\end{eqnarray*}

Since the optimal value for $\lambda$ is equal to the least squares solution of $Z$ over $V_kD_k$, it follows that  $\|Z\lambda \|_F^2=\|Z(Z^TZ)^{-1}Z^TV_kD_k\|_F^2 = \|P_Z V_kD_k\|_F^2$, and the proof is complete. 
\end{proof}


\subsection{Sparse \textsc{og} procedure}
\label{app:alg}
The following pseudo-code illustrates the key-steps to implement the \textsc{sog} procedure illustrated in Section~\ref{sec:sog}. An \textsc{r} implementation is available at \url{github.com/emanuelealiverti/sog}.

\begin{algorithm*}[H]
	\caption{\textsc{sog} algorithm}
	\KwIn{Data matrix $X$, $n\times p$. Approximation rank $k$.}
	\For{$j=1,\dots,k$}{
		\While{Changes in $u_j$ and $s_j$ are not sufficiently small}{
	Compute $\beta_j$ via least squares as  $$\beta_j = (Z^TZ)^{-1}Z^TP_{j-1}\X u_j,$$ with $P_{j-1} = I_{n\times n} - \sum_{l=1}^{j-1} s_l s_l^T$

	Update $ s_j \in \mathbb{R}^n$ as $$ s_j = \frac{P_{j-1}\X u_j - \beta_j Z}{\|P_{j-1} \X u_j  -\beta_j Z\|_2}$$

 Update $ u_j \in \mathbb{R}^p$ as $$ u_j = \frac{\mathcal{S}_\theta (\X^T s_j)}{\|\mathcal{S}_\theta(\X^T s_j)\|_2},$$
where $\mathcal{S}_\theta(x) = \text{sign}(x)(|x| -\theta)\mathbb{I}(|x|\geq \theta)$ and

\eIf{$\|\X^T s_j\|_1\leq t$}{Set $\theta = 0$}{
		Set $\theta > 0$ such that $\|u_j\|_1 = t$
	}
	}
	}
	\KwOut{Set $d_j =  s_j^T\X u_j$. Let $\Sb$ denote the $n \times k$ matrix with columns $d_j s_j$, $j=1,\dots,k$.
	Let $\U$ denote the $p\times k$ sparse matrix with rows $u_j$, $j=1,\dots,k$. Return $\widetilde X = SU^T$}
\end{algorithm*}


\begingroup
    \fontsize{11pt}{12pt}\selectfont
\bibliographystyle{rss}
\bibliography{short.bib}
\endgroup

\end{document}